\documentclass[conference,letterpaper,romanappendices]{ieeeconf}
\let\proof\relax   

\usepackage{amsthm,xpatch}
\usepackage{amsmath}
\usepackage{cite}
\usepackage{amssymb}
\usepackage{amsfonts}
\usepackage{dsfont}
\usepackage{graphicx, subfigure}
\usepackage{color}
\usepackage{breqn}
\usepackage{mathtools}
\usepackage{bbm}
\usepackage{latexsym}
\usepackage[ruled, linesnumbered]{algorithm2e}
\usepackage{accents}

\newtheorem{lemma}{Lemma}
\newtheorem{theorem}{Theorem}

\newtheorem{result}{Result}

\IEEEoverridecommandlockouts

\begin{document}

\newcommand{\SB}[3]{
\sum_{#2 \in #1}\biggl|\overline{X}_{#2}\biggr| #3
\biggl|\bigcap_{#2 \notin #1}\overline{X}_{#2}\biggr|
}

\newcommand{\Mod}[1]{\ (\textup{mod}\ #1)}

\newcommand{\overbar}[1]{\mkern 0mu\overline{\mkern-0mu#1\mkern-8.5mu}\mkern 6mu}

\makeatletter
\newcommand*\nss[3]{%
  \begingroup
  \setbox0\hbox{$\m@th\scriptstyle\cramped{#2}$}%
  \setbox2\hbox{$\m@th\scriptstyle#3$}%
  \dimen@=\fontdimen8\textfont3
  \multiply\dimen@ by 4             
  \advance \dimen@ by \ht0
  \advance \dimen@ by -\fontdimen17\textfont2
  \@tempdima=\fontdimen5\textfont2  
  \multiply\@tempdima by 4
  \divide  \@tempdima by 5          
  \ifdim\dimen@<\@tempdima
    \ht0=0pt                        
    \@tempdima=\fontdimen5\textfont2
    \divide\@tempdima by 4          
    \advance \dimen@ by -\@tempdima 
    \ifdim\dimen@>0pt
      \@tempdima=\dp2
      \advance\@tempdima by \dimen@
      \dp2=\@tempdima
    \fi
  \fi
  #1_{\box0}^{\box2}%
  \endgroup
  }
\makeatother

\makeatletter
\renewenvironment{proof}[1][\proofname]{\par
  \pushQED{\qed}%
  \normalfont \topsep6\p@\@plus6\p@\relax
  \trivlist
  \item[\hskip\labelsep
        \itshape
    #1\@addpunct{:}]\ignorespaces
}{%
  \popQED\endtrivlist\@endpefalse
}
\makeatother

\makeatletter
\newsavebox\myboxA
\newsavebox\myboxB
\newlength\mylenA

\newcommand*\xoverline[2][0.75]{%
    \sbox{\myboxA}{$\m@th#2$}%
    \setbox\myboxB\null
    \ht\myboxB=\ht\myboxA%
    \dp\myboxB=\dp\myboxA%
    \wd\myboxB=#1\wd\myboxA
    \sbox\myboxB{$\m@th\overline{\copy\myboxB}$}
    \setlength\mylenA{\the\wd\myboxA}
    \addtolength\mylenA{-\the\wd\myboxB}%
    \ifdim\wd\myboxB<\wd\myboxA%
       \rlap{\hskip 0.5\mylenA\usebox\myboxB}{\usebox\myboxA}%
    \else
        \hskip -0.5\mylenA\rlap{\usebox\myboxA}{\hskip 0.5\mylenA\usebox\myboxB}%
    \fi}
\makeatother

\xpatchcmd{\proof}{\hskip\labelsep}{\hskip3.75\labelsep}{}{}

\pagestyle{plain}

\title{\fontsize{22.59}{28}\selectfont A Systematic Approach to Incremental\\Redundancy over Erasure Channels}

\author{Anoosheh Heidarzadeh, Jean-Francois Chamberland,\thanks{A. Heidarzadeh and J.-F. Chamberland are with the Department of Electrical and Computer Engineering, Texas A\&M University, College Station, TX 77843 USA (E-mail: \{anoosheh, chmbrlnd\}@tamu.edu).} Parimal Parag,\thanks{P. Parag is with the Department of Electrical Communication Engineering, Indian Institute of Science, Bengaluru, India (E-mail: parimal@iisc.ac.in).} and Richard D. Wesel\thanks{R.D. Wesel is with the Department of Electrical Engineering, University of California, Los Angeles, CA 90095 USA (E-mail: wesel@ucla.edu).}\thanks{This material is based upon work supported by the National Science Foundation under Grants No.~CCF-1619085 and~CCF-1618272.}}


\maketitle

\thispagestyle{plain}

\begin{abstract}
As sensing and instrumentation play an increasingly important role in systems controlled over wired and wireless networks, the need to better understand delay-sensitive communication becomes a prime issue.
Along these lines, this article studies the operation of data links that employ incremental redundancy as a practical means to protect information from the effects of unreliable channels.
Specifically, this work extends a powerful methodology termed sequential differential optimization to choose near-optimal block sizes for hybrid ARQ over erasure channels.
In doing so, an interesting connection between random coding and well-known constants in number theory is established.
Furthermore, results show that the impact of the coding strategy adopted and the propensity of the channel to erase symbols naturally decouple when analyzing throughput.
Overall, block size selection is motivated by normal approximations on the probability of decoding success at every stage of the incremental transmission process.
This novel perspective, which rigorously bridges hybrid ARQ and coding, offers a pragmatic means to select code rates and blocklengths for incremental redundancy.
\end{abstract}

\section{Introduction}

As the reach of the Internet keeps expanding beyond its traditional applications and incorporates more sensing, actuation, and cyber-physical systems, there is a pressing need to better understand delay-sensitive communication over unreliable channels.
The increasing popularity of interactive communications, live gaming over mobile devices, and augmented reality also contribute to a growing interest in low-latency connections.
These circumstances have been an important motivating factor underlying several recent inquiries pertaining to information transfers under stringent delay constraints.
Such contributions include the divergence framework for short blocklengths~\cite{polyanskiy2010dispersion,polyanskiy2011dispersion}, the interplay between coding and queueing~\cite{ieee-tit-2013-kcp}, and ongoing work on the age of information~\cite{ephremides2016age,yates2017age}.

Hybrid automatic repeat request (ARQ) has been identified as a key approach to deliver information in a timely manner over unreliable channels.
It can be designed to adapt gracefully to channel degradations associated with fading and interference,
and it has found wide application in theory and practice.
In some sense, hybrid ARQ is a means to leverage limited feedback between a source and its destination to ensure the timely delivery of information, especially in short blocklength regimes.
Researchers have developed techniques to analyze the benefits of communication systems with hybrid ARQ \cite{le2007queueing,ieee-tcomm-2015-hcp}. Yet, until recently, brute force searches, simulation studies, and ad-hoc schemes remained the primary means of parameter selection in terms of blocklengths and code rate for such systems, see, e.g.,~\cite{wesel2015harq}. This situation changed when Vakilinia et al.\ introduced a novel approach for parameter selection~\cite{VWRDW:2014,VRDW:2016}.
Their proposed methodology captures the effects of the physical channel on code performance by defining an approximate empirical distribution on the probability that a rate compatible code decodes successfully at each of its available rates.
Based on the ensuing distribution, the authors then put forth a numerically efficient, sequential differential optimization (SDO) algorithm that yields best operational parameters for hybrid ARQ.

In this work, we extend the algorithm of~\cite{VRDW:2016} to erasure channels, and we characterize performance for a class of random codes. Our results provide an algorithmic blueprint for parameter selection applied to hybrid ARQ and random codes, a popular combination in the literature. The analysis reveals a clear separation between the effects of the erasure channel and the attributes of the underlying code in selecting block sizes. For the adopted coding scheme, a systematic approach that links decoding success to the number of observed symbols is derived based on moment matching. Altogether, the performance of a system with incremental redundancy hinges on three main components: the coding scheme employed, the behavior of the erasure channel, and the quantization effects associated with hybrid ARQ blocks.
These components, along with design decisions, are discussed below.


\section{Coding Scheme and Incremental Redundancy}

We assume that the communication system employs forward error correction to protect information bits from potential channel erasures, as in \cite{ieee-tit-2013-kcp}.
We denote the size of the original message by $k$.
Redundancy is added using a random coding scheme, which serves as an analytically tractable proxy for more practical codes \cite{Gallager0471290483,Richardson0521852293}.
Specifically, the encoding of a message proceeds as follows.
First, a random binary parity-check matrix of size $(n-k) \times n$ is generated, with every entry selected uniformly over the binary alphabet, independently from other elements.
The nullspace of this matrix yields a codebook for the transmission.
The mapping of a message to a codeword is then performed using an arbitrary choice function known to both the source and the destination. Maximum-likelihood decoding is performed at the destination to recover the original message. This coding scheme is known to perform well for large $n$, and it enables the fine selection of a code rate, as any rate of the form $k/n$, where $k \leq n$, is admissible.

Under hybrid ARQ, encoded symbols are transmitted in distinct sub-blocks, rather than all at once. Untransmitted bits can simply be labeled as erasures during early decoding attempts.
Furthermore, the statistical symmetry in the random code structure ensures that the probability of decoding success only depends on the number of erased symbols, rather than their specific locations.
In this context, one can examine the effective blocklength after the transmission of $j$ sub-blocks or, equivalently, one can focus on the size of sub-block $j$. Throughout, we represent the effective blocklength upon transmission of $j$ sub-blocks by $n_j$, and we write $l_j$ to denote the number of symbols transmitted as part of step~$j$.
These quantities are related via the equations $n_j = \sum_{i=1}^j l_i$ and $l_j = n_{j} - n_{j-1}$ for all $j\in [m]$, where, for convenience, we adopt the convention $n_0 = 0$.
Consider the use of such a random code paired to a memoryless binary erasure channel with erasure probability $0 \leq \epsilon < 1$.
Let $\mathbf{c} = (c_1, \dots, c_n)$ denote the codeword corresponding to a given message $\mathbf{x} = (x_1, \dots, x_k)$.
Also, let $\mathbf{c}_i \triangleq \{c_j: n_{i-1} < j \leq n_{i}\}$ for $i \in [m]$ be the $i$th sub-block of this codeword; consequently, the length of sub-block~$i$ is $|\mathbf{c}_i| = l_i$. 

Each transmission round proceeds as follows.
First, the source sends sub-block $\mathbf{c}_1$.
The destination then receives sub-block $\mathbf{c}_1$, or a proper subset thereof, depending on the realized erasure pattern.
The destination seeks to recover the message $\mathbf{x}$, and responds by sending an ACK or NACK to the source (over an erasure/error-free feedback channel) based on the outcome the decoding attempt.
If the source receives a NACK, it continues with the transmission of sub-block $\mathbf{c}_2$, and waits for an ACK or NACK.
This action repeats until (i) the source receives an ACK and forgoes the transmission of the remaining sub-blocks, if any; or (ii) the source exhausts all the available sub-blocks for this codeword.
In case (i), the transmission round is deemed successful, and the source proceeds to the next message; whereas, in case (ii), the transmission round fails.
In the event of a failure, the destination discards the information aggregated since the beginning of the round, and the source starts a new transmission round for the message $\mathbf{x}$.
Within the framework where a maximum of $m$ sub-blocks are available for retransmission, our goal is to identify sub-block sizes $\{l_i\}_{i=1}^{m-1}$ such that the expected number of transmitted symbols until the message becomes decodable is minimized. 


\section{Asymptotic Analysis of Random Codes}

For the random coding scheme at hand, we use $P_{\mathrm{s}}(k,n,r)$ to represent the probability of decoding success as a function of the number of symbols available at the receiver, $r$.
We note that the probability of decoding failure as a function of $r$ necessarily becomes $P_{\mathrm{f}}(k,n,r) \triangleq 1-P_{\mathrm{s}}(k,n,r)$.
The following result holds. 
(The proofs of all lemmas in this section can be found in the appendix.)

\begin{lemma}\label{lem:DSP}
The probability of decoding success for the aforementioned random code is
\begin{equation}
\label{equation:DecodingSuccessExtended}
P_{\mathrm{s}} (k, n, r)
\triangleq \begin{cases}
0, & r < k \\
\prod_{l=0}^{n - r - 1} \left( 1 - 2^{l - (n-k)} \right), & k \leq r \leq n \\
1, & r > n .
\end{cases}
\end{equation} 
\end{lemma}

Although the number of sent symbols and, consequently, the number of symbols available at the receiver cannot exceed the blocklength in practice, we extend function $P_{\mathrm{s}}(k,n,r)$ in \eqref{equation:DecodingSuccessExtended} to cases where $r > n$.
The utility of this definition will become manifest shortly, when comparing coding schemes. 

\begin{lemma} \label{lemma:StochaticDominance}
For any $k$ and $r$, the function $P_{\mathrm{s}} (k, n, r)$ is monotone decreasing in $n$.
\end{lemma}

Let $R_t$ be the number of observed (non-erased) symbols available to the receiver at time~$t$, and let $P_{R_t}$ be the discrete probability measure associated with $R_t$.
Then, $R_t$ has a binomial distribution with parameters $t$ and $1-\epsilon$, i.e., $P_{R_t}(r) = \binom{t}{r} \epsilon^{t-r} (1-\epsilon)^{r}$.
Moreover, the probability that the destination sends an ACK at time $t$ or earlier becomes
\begin{equation*}\label{eq:Packni}
P_{\mathrm{ack}}(t) \triangleq 1 - \textstyle \sum_{e=0}^{t} P_{\mathrm{f}}(k,n,t-e) P_{R_t}(t-e)
\end{equation*}
for $k\leq t\leq n$; $P_{\mathrm{ack}}(t) = 0$ for $t<k$. 
Consequently, we get $P_{\mathrm{nack}}(t)\triangleq 1-P_{\mathrm{ack}}(t)$.
In words, $P_{\mathrm{ack}}(n_i)$ and $P_{\mathrm{nack}}(n_i)$ designate the probabilities that, after the transmission of the $i$th block, the source receives an ACK and a NACK, respectively. 
 
Let $S\in [m]$ be the index of the last sub-block that the source sends within a transmission round, and let $n_S$ be the corresponding number of symbols.
Note that $S$ and $n_S$ are random variables; the expected number of symbols, $\mathbb{E}[n_S]$, is given by 
\begin{equation} \label{eq:Ens}
\begin{split}
\mathbb{E}[n_S]
&= \textstyle \sum_{i=2}^{m} n_i\left(P_{\mathrm{ack}}(n_i)-P_{\mathrm{ack}}(n_{i-1})\right) \\
&\quad +n_1P_{\mathrm{ack}}(n_1) + n_m P_{\mathrm{nack}}(n_m) \\
&= \textstyle \sum_{i=1}^{m-1} (n_i-n_{i+1})P_{\mathrm{ack}}(n_i)+n_m .
\end{split}
\end{equation}
Expectation $\mathbb{E}[n_S]$ is a multivariate function of $\{n_i\}_{i\in [m-1]}$. For any given $k,n,m$, and $\epsilon$, the problem is to find $\{n_i\}_{i\in [m-1]}$ such that $\mathbb{E}[n_S]$ is minimum.

\subsection{Asymptotic Behavior over Reliable Channels}

We initiate our analysis by focusing on the special case of a lossless channel, i.e., $\epsilon=0$.
We consider an elementary version of the problem where the transmitted symbols are received in sequence (not in blocks), and a decoding attempt takes place every time a new symbol arrives (not only after the reception of a new sub-block). 
For  $k$ and $n$ fixed, let $M_n$ be a random variable that denotes the number of symbols needed for the message to become decodable, following the prescribed order.
Under this definition, we get $k\leq M_n\leq n$ and $\Pr(M_n\leq r) = P_{\mathrm{s}}(k,n,r)$.
We wish to study the asymptotic behavior of the mean and variance of $M_n$ as $n$ grows unbounded.

Denote the Erd\"{o}s-Borwein constant by $c_0 = \sum_{i=1}^{\infty} \frac{1}{2^{i}-1}=1.6066951524...$, and the digital search tree constant by $c_1 = \sum_{i=1}^{\infty} \frac{1}{(2^{i}-1)^2}=1.1373387363...$. The following infinite sums of products are key components in our analysis.   

\begin{lemma}\label{lem:infsums}
For $a_i \triangleq 2^{-i}\prod_{j=i+1}^{\infty} \left( 1-2^{-j} \right)$, it holds that
\begin{align*}
&\textstyle \sum_{i=0}^{\infty} a_i = 1 \\	
&\textstyle \sum_{i=0}^{\infty} i a_i = c_0 =1.6066951524... \\
&\textstyle \sum_{i=0}^{\infty} i^2 a_i = c^2_0+c_0+c_1 =5.3255032015...
\end{align*}
\end{lemma}

Let $P_{M_n}$ be the discrete probability measure for $M_n$, i.e., $P_{M_n}(r)=P_{\mathrm{s}}(k,n,r)-P_{\mathrm{s}}(k,n,r-1)$.
It is easy to verify that $P_{M_n}(r)=2^{k-r}P_{\mathrm{s}}(k,n,r)=2^{k-r} \prod_{l=0}^{n-r-1} \left( 1-2^{l-(n-k)} \right)$ for $k\leq r\leq n$; and $P_{M_n}(r)=0$ for $r<k$ and $r>n$.
We emphasize that $\sum_{r=k}^{n} P_{M_n}(r)=1$.
Thus, the mean of $M_n$ is given by
\begin{equation*}
\mathbb{E}[M_n]= \sum_{r=k}^n r P_{M_n} (r)
= \sum_{i=0}^{n-k} (k+i) 2^{-i}\prod_{j=i+1}^{n-k} \left( 1-2^{-j} \right).
\end{equation*}
Similarly, the second moment of $M_n$ is given by
\begin{equation*}
\mathbb{E}[M_n^2]= \sum_{r=k}^n r^2 P_{M_n} (r)
= \sum_{i=0}^{n-k} (k+i)^2 2^{-i}\prod_{j=i+1}^{n-k} \left( 1-2^{-j} \right).
\end{equation*}

\begin{theorem}\label{thm:meanvar}
For any $k$,
\begin{align}
\label{equation:MeanEM}
\lim_{n\rightarrow\infty}\mathbb{E}[M_n] &= k + c_0 \\
\label{equation:VarianceEM}
\lim_{n\rightarrow\infty}\mathrm{Var}[M_n] &= c_0 + c_1 .	
\end{align}
\end{theorem}
\begin{proof}
Taking the limit, as $n$ becomes large, we get
\begin{align*}
\lim_{n \rightarrow \infty} \mathbb{E}[M_n]
&= \textstyle \sum_{i=0}^{\infty} (k+i) 2^{-i} \prod_{j=i+1}^{\infty} (1 - 2^{-j}) \\ 
\lim_{n\rightarrow\infty} \mathbb{E}[M_n^2]
&= \textstyle \sum_{i=0}^{\infty} (k^2 + 2ki + i^2)2^{-i} \prod_{j=i+1}^{\infty} \left( 1 - 2^{-j} \right) .
\end{align*}
By Lemma~\ref{lem:infsums}, we have \[\lim_{n \rightarrow \infty} \mathbb{E}[M_n]= k+c_0.\] Similarly, we have \[\lim_{n\rightarrow\infty} \mathbb{E}[M_n^2]= k^2+2kc_0+(c_0^2+c_0+c_1).\]
Also, since the variance of $M_n$ can readily be computed as $\mathrm{Var}[M_n]=\mathbb{E}[M^2_n]-\mathbb{E}[M_n]^2$, we obtain 
\begin{align*}
\lim_{n\rightarrow\infty} \mathrm{Var}[M_n] & = k^2+2kc_0+(c^2_0+c_0+c_1)-(k+c_0)^2\\ & =c_0+c_1.	
\end{align*} This completes the proof.
\end{proof}

\subsection{Asymptotic Behavior over Unreliable Channels}

We turn to the more elaborate problem whereby symbols are transmitted over a lossy channel.
That is, individual symbols are erased with probability $\epsilon > 0$.
For $k$, $n$, and $\epsilon$ fixed, we represent the length of a communication round by $N_n$.
Note that $k\leq N_n\leq n$.
We can partition rounds into two categories: (i) the receiver can decode before all the symbols are transmitted, and $N_n$ corresponds to the first instant at which the message can be successfully recovered; (ii) all the symbols are exhausted during the transmission phase, and $N_n = n$ irrespective of the outcome of the decoding process.
Again, we wish to study the asymptotic behavior of the mean and variance of $N_n$ as $n$ increases to infinity.

Define $E_r$ as the number of symbols erased prior to observing the $r$th unerased symbols at the receiver.
We write $P_{E_r}$ to denote the  discrete probability measure associated with $E_r$, and we note that this random variable possesses a negative binomial distribution with parameters $r$ and $\epsilon$.
In other words, we have $P_{E_r}(e) = \binom{r+e-1}{e} \epsilon^{e} (1-\epsilon)^{r}$ for $e\geq 0$.
Then, $\Pr(N_n = t)=\sum_{r=k}^{t} P_{E_r}(t-r) P_{M_n}(r)$ for $k\leq t<n$, and $\Pr(N_n = n) = 1-\sum_{t=k}^{n-1} \Pr(N_n = t)=\sum_{t=n}^{\infty}\sum_{r=k}^{t} P_{E_r}(t-r) P_{M_n}(r)$.
Thus, we can write 
\begin{equation*}
\begin{split}
\mathbb{E}[N_n] &= \textstyle \sum_{t=k}^{n} t\Pr(N_n = t) \\
&= \textstyle \sum_{t=k}^{\infty} \sum_{r=k}^{t} \min(t,n) P_{E_r}(t-r)P_{M_n}(r)\\ 
&= \textstyle \sum_{r=k}^{n}\sum_{e=0}^{\infty}\min(r+e,n)P_{E_r}(e)P_{M_n}(r).
\end{split}	
\end{equation*} Similarly, we can write
\begin{equation*}
\mathbb{E}[N^2_n] = \sum_{r=k}^{n}\sum_{e=0}^{\infty}\min((r+e)^2,n^2)P_{E_r}(e)P_{M_n}(r).
\end{equation*} 

\begin{theorem}\label{thm:meanvargen}
For any $k$ and $\epsilon$,
\begin{align}
\label{equation:MeanEN}
\mu(k,\epsilon) &\triangleq \lim_{n\rightarrow\infty}\mathbb{E}[N_n] = \frac{k+c_0}{1-\epsilon} \\\label{equation:VarianceEN}
\sigma^2(k,\epsilon) &\triangleq \lim_{n\rightarrow\infty}\mathrm{Var}[N_n] = \frac{(k+c_0)\epsilon+c_0+c_1}{(1-\epsilon)^2} .
\end{align}
\end{theorem}
\begin{proof}
Observing that $\min(r+e,n) \leq r+e$, we get
\begin{equation*}
\mathbb{E}[N_n] \leq \textstyle \sum_{r=k}^{n} \sum_{e=0}^{\infty}(r+e) P_{E_r}(e)P_{M_n}(r)
\end{equation*}
for all $n$.
For a memoryless erasure channel, $E_r$ possesses a negative binomial distribution with parameters $r$ and $\epsilon$ and
\begin{equation*}
\begin{split}
\textstyle \sum_{e=0}^{\infty} (r+e)P_{E_r}(e)
&= \textstyle r \sum_{e=0}^{\infty} P_{E_r}(e) + \sum_{e=0}^{\infty}eP_{E_r}(e) \\
&= r+\mathbb{E}[E_r]
=  r/(1-\epsilon)
\end{split}
\end{equation*}
for all $r$. Thus, for any $n$, we have
\begin{equation}\label{eq:RHSEU}
\begin{split}
\mathbb{E}[N_n] &\leq \frac{1}{1-\epsilon} \textstyle \sum_{r=k}^{n} r P_{M_n}(r)
= \displaystyle \frac{\mathbb{E}[M_n]}{1-\epsilon} .
\end{split}
\end{equation}
Next, we establish a lower bound for $\mathbb{E}[N_n]$;
\begin{equation*}
\mathbb{E}[N_n]\geq \textstyle \sum_{r=k}^{n}\sum_{e=0}^{n-r}(r+e) P_{E_r}(e)P_{M_n}(r).	
\end{equation*}
Since $P_{\mathrm{s}}(k,n,r)$ is monotone decreasing in $n$ and $P_{M_n}(r) = 2^{k-r}P_{\mathrm{s}}(k,n,r)$ for all $k\leq r\leq n$, we gather that $P_{M_n}(r)$ is monotone decreasing in $n$ for all $k\leq r\leq n$.
This implies that $P_{M_n}(r)\geq \lim_{n\rightarrow\infty} P_{M_n}(r)$ for all $n$ and all $k\leq r\leq n$. Note that $\lim_{n\rightarrow\infty} P_{M_n}(r) = 2^{k-r}\prod_{j=r-k+1}^{\infty}(1-2^{-j})$. Thus, for all $n$, we can write 
\begin{equation} \label{eq:RHSEL}
\begin{split}
\mathbb{E}[N_n] &\geq \sum_{r=k}^{n} \sum_{e=0}^{n-r} (r+e) P_{E_r}(e) 2^{k-r} \prod_{j=r-k+1}^{\infty}(1-2^{-j})\\
&=\sum_{r=k}^{n} 2^{k-r}  \sum_{e=0}^{n-r} (r+e)P_{E_r}(e) \prod_{j=r-k+1}^{\infty}(1-2^{-j}).	
\end{split}
\end{equation}
Using Theorem~\ref{thm:meanvar}, we deduce that the RHS of~\eqref{eq:RHSEU} converges to $(k+c_0)/(1-\epsilon)$ as $n$ grows unbounded.
Furthermore, in view of Lemma~\ref{lem:infsums}, we see that the RHS of~\eqref{eq:RHSEL} converges to 
\begin{equation*}
\begin{split}
&\textstyle \sum_{r=k}^{\infty} 2^{k-r}\sum_{e=0}^{\infty}(r+e)P_{E_r}(e) \prod_{j=r-k+1}^{\infty}(1-2^{-j}) \\
&= \frac{1}{1-\epsilon} \textstyle \sum_{r=k}^{\infty} r 2^{k-r} \prod_{j=r-k+1}^{\infty} (1-2^{-j}) \\
&= \frac{1}{1-\epsilon} \textstyle \sum_{i=0}^{\infty} (k+i) 2^{-i}\prod_{j=i+1}^{\infty} (1-2^{-j})
= \displaystyle \frac{k+c_0}{1-\epsilon} .
\end{split}
\end{equation*}
Combining~\eqref{eq:RHSEU} and~\eqref{eq:RHSEL}, the sandwich theorem yields~\eqref{equation:MeanEN}.

By adopting a similar approach, we can obtain upper bound
\begin{equation}\label{eq:RHSEEU}
\begin{split}
\mathbb{E} \left[ N^2_n \right]
&\leq \textstyle \sum_{r=k}^{n}\sum_{e=0}^{\infty} (r+e)^2 P_{E_r}(e) P_{M_n}(r) \\
&= \frac{\mathbb{E}[M^2_n]+\epsilon \mathbb{E}[M_n]}{(1-\epsilon)^2}
\end{split}
\end{equation}
and lower bound
\begin{equation}\label{eq:RHSEEL}
\begin{split}
\mathbb{E} & \left[ N^2_n \right]
\geq \sum_{r=k}^{n}\sum_{e=0}^{n-r} (r+e)^2 P_{E_r}(e) P_{M_n}(r) \\
&\geq \sum_{r=k}^{\infty} 2^{k-r}\sum_{e=0}^{\infty}(r+e)^2 P_{E_r}(e) \prod_{j=r-k+1}^{\infty}(1-2^{-j})
\end{split}
\end{equation}
for all $n$.
As $n$ goes to infinity, the RHS of~\eqref{eq:RHSEEU} converges to $\left( (k+c_0)^2+(k+c_0)\epsilon+c_0+c_1 \right)/(1-\epsilon)^2$ by Theorem~\ref{thm:meanvar}.
Likewise, by Lemma~\ref{lem:infsums}, the RHS of~\eqref{eq:RHSEEL} converges to 
\begin{equation*}
\begin{split}
&\frac{1}{(1-\epsilon)^2} \textstyle \sum_{r=k}^{\infty} r(r+\epsilon) 2^{k-r} \prod_{j=r-k+1}^{\infty} (1-2^{-j}) \\
&= \frac{1}{(1-\epsilon)^2} \textstyle \sum_{i=0}^{\infty} (k+i)^2 2^{-i} \prod_{j=i+1}^{\infty} (1-2^{-j}) \\
&\quad + \frac{\epsilon}{(1-\epsilon)^2} \textstyle \sum_{i=0}^{\infty} (k+i) 2^{-i} \prod_{j=i+1}^{\infty} (1-2^{-j}) \\
&= \left( (k+c_0)^2+(k+c_0)\epsilon+c_0+c_1 \right)/(1-\epsilon)^2 .
\end{split}
\end{equation*}
Combining~\eqref{eq:RHSEEU} and~\eqref{eq:RHSEEL}, the sandwich theorem offers a tight characterization of the asymptotic second moment of $N_n$,
\begin{equation*}
\lim_{n\rightarrow\infty}\mathbb{E}[N^2_n]	
=  \frac{ (k+c_0)^2+(k+c_0)\epsilon+c_0+c_1 }{(1-\epsilon)^2} .
\end{equation*}
From its first and second moments, we can infer the variance of $N_n$;
$\lim_{n\rightarrow\infty} \mathrm{Var}[N_n] = \lim_{n\rightarrow\infty} \mathbb{E}[N^2_n]-\lim_{n\rightarrow\infty} \mathbb{E}[N_n]^2 = ((k+c_0)\epsilon+c_0+c_1)/(1-\epsilon)^2$.
\end{proof}

\section{Sequential Differential Optimization}

In this section, we apply sequential differential optimization (SDO) to incremental redundancy in the context of random codes over erasure channels~\cite{VRDW:2016}.
This technique works with continuous random variables and, as such, we must find a suitable approximation for the distribution of discrete random variable $N_n$.
To begin, we emphasize that $P_{\mathrm{ack}}$ matches the cumulative distribution function (CDF) of $N_n$ for $t < n$.
Below, we adopt moment matching to find a smooth approximation to $P_{\mathrm{ack}}$~\cite{casella2001statistical}.
The mean and variance of $N_n$, in the asymptotic regime where $n$ becomes large, are given in Theorem~\ref{thm:meanvargen}.
We can therefore approximate $P_{\mathrm{ack}}$ by the CDF of a  real-valued random variable $X$, which we denote by $F$, where $\mathbb{E}[X] = \mu(k,\epsilon)$ and $\mathrm{Var}[X] = \sigma^2(k,\epsilon)$.
Henceforth, we replace $\mu(k,\epsilon)$ and $\sigma^2(k,\epsilon)$ by the implicit forms $\mu$ and $\sigma^2$, respectively, for notational convenience.

The CDF approximation enables us to utilize SDO to find near optimal values for sub-block sizes $\{n_i\}$.
This technique works as follows. 
Given $\{n_j\}_{1\leq j<i}$, the optimal value of $n_i$, $1<i<m$, can be computed via setting the partial derivative of $\tilde{\mathbb{E}}[n_S]$, which is defined as $\mathbb{E}[n_S]$ when $P_{\mathrm{ack}}$ is replaced by $F$, with respect to $n_{i-1}$ to zero and solving for $n_i$. The structure of SDO immediately yields the following result.

\begin{result}\label{lem:SDO}
For any $1< i<m$, an approximation of the optimal value of $n_i$, denoted by $\tilde{n}_i$, as a function of $\{\tilde{n}_j\}_{1\leq j<i}$, is obtained recursively via
\begin{equation*}
\tilde{n}_{i} =
\tilde{n}_{i-1}+\left\lceil\left(F(\tilde{n}_{i-1})-F(\tilde{n}_{i-2})\right) \left(\Bigl.\frac{dF(x)}{dx}\Bigr|_{x=\tilde{n}_{i-1}}\right)^{-1}\right\rceil
\end{equation*}
where $\tilde{n}_1$ is given and $n_0 \triangleq -\infty$. 
\end{result}
\begin{proof} 
Since $\tilde{\mathbb{E}}[n_S] = \sum_{i=1}^{m-1} (n_i-n_{i+1})F(n_i)+n_m$, 
we have
\[\frac{\partial\tilde{\mathbb{E}}[n_S]}{\partial n_1}
= F(n_1)+(n_1-n_2)\Bigl.\frac{dF(x)}{dx}\Bigr|_{x=n_1}.\]
Setting $\frac{\partial\tilde{\mathbb{E}}[n_S]}{\partial n_1}=0$ and solving for $n_2$ yields the result for $i=2$. Similarly, it can be seen that
\[\frac{\partial\tilde{\mathbb{E}}[n_S]}{\partial n_i}=F(n_i)-F(n_{i-1})+(n_{i}-n_{i+1})\Bigl.\frac{dF(x)}{dx}\Bigr|_{x=n_i}.\]
Setting $\frac{\partial\tilde{\mathbb{E}}[n_S]}{\partial n_i}=0$ yields the result for all $2<i<m$.
\end{proof}

\paragraph*{Normal Approximation} 
Paralleling the steps in \cite{VRDW:2016}, we first approximate the distribution of $N_n$ by a normal distribution $\mathcal{N}(\mu,\sigma^2)$ with mean $\mu$ and variance $\sigma^2$, as defined above. The CDF of $N_n$ is then approximated by $F(x)\triangleq 1-Q(\frac{x-\mu}{\sigma})$, where $Q(x) = \frac{1}{\sqrt{2 \pi}}\int_{x}^{\infty} e^{-t^2/2} dt$ is the complementary CDF of a standard Gaussian variable.
Note that $\frac{dF(x)}{dx}=-\frac{1}{\sigma}Q'(\frac{x-\mu}{\sigma})$ where $Q'(x)\triangleq -\frac{1}{\sqrt{2\pi}}e^{-x^2/2}$.  

For any $k$, $m$, and $\epsilon$, an approximation of the optimal value of $n_i$, denoted by $\tilde{n}_i$, is given by 
\begin{equation*}
\begin{split}
\textstyle \tilde{n}_{i} = \tilde{n}_{i-1}
+ & \left\lceil \left( Q \left(\frac{\tilde{n}_{i-1}-\mu}{\sigma}\right)-Q\left(\frac{\tilde{n}_{i-2}-\mu}{\sigma}\right)\right) \right. \\
&\quad \times \left. \left(\frac{1}{\sigma}Q'\left(\frac{\tilde{n}_{i-1}-\mu}{\sigma}\right)\right)^{-1}\right\rceil
\end{split}
\end{equation*}
for all $1<i<m$, where $\tilde{n}_1$ is given, and $\tilde{n}_0\triangleq -\infty$. 

\paragraph*{Log-Normal Approximation} 
An alternate candidate approximation for $P_{\mathrm{ack}}$ is given by the log-normal distribution $\mathcal{LN}(\mu_{\star},\sigma^2_{\star})$ with parameters $\mu_{\star}= \ln(\mu^2/\sqrt{\mu^2+\sigma^2})$ and $\sigma^2_{\star}=\ln(1+\sigma^2/\mu^2)$.
That is, the CDF of $N_n$ can be approximated by $F_{\star}(x)\triangleq 1-Q ( \frac{\ln x-\mu_{\star}}{\sigma_{\star}})$. Note that, in this case, $\frac{dF_{\star}(x)}{dx} = -\frac{1}{x\sigma_{\star}}Q' ( \frac{\ln x-\mu_{\star}}{\sigma_{\star}})$.

For any $k$, $m$, and $\epsilon$, an approximation of the optimal value of $n_i$, denoted by $\tilde{n}_i$, is given by 
\begin{equation*}
\begin{split}
\textstyle \tilde{n}_{i} = \tilde{n}_{i-1}
+ & \left\lceil \left( Q\left(\frac{\ln \tilde{n}_{i-1}-\mu_{\star}}{\sigma_{\star}}\right)-Q\left(\frac{\ln \tilde{n}_{i-2}-\mu_{\star}}{\sigma_{\star}}\right) \right) \right. \\
&\quad \times \left. \left(\frac{1}{\tilde{n}_{i-1} \sigma_{\star}} Q'\left(\frac{\ln \tilde{n}_{i-1}-\mu_{\star}}{\sigma_{\star}}\right)\right)^{-1}\right\rceil
\end{split}
\end{equation*}
for all $1<i<m$, where $\tilde{n}_1$ is given, and $\tilde{n}_0\triangleq 0$.

\section{Results}
In this section, we first compare the performance of a system where parameters are obtained via exhaustive search to that of a competing implementation with parameters tuned using SDO, both with the normal approximation (NA) and the log-normal approximation (LNA).
Figure~\ref{fig:TvsKESNA} depicts throughput $T \triangleq kP_{\mathrm{ack}}(n)/\mathbb{E}[n_S]$ versus  message size $k$ for $\epsilon=0.5$, $n\in \{88,104\}$, and $m\in \{2,4\}$.
As can be seen, both SDO-NA and SDO-LNA result in near optimal throughput, with the performance curves being essentially indistinguishable.
For scenarios beyond those shown in the figure, our extensive numerical evaluations suggest that SDO-LNA slightly outperforms SDO-NA.  

\begin{figure}
\centering{\includegraphics[width=0.45\textwidth]{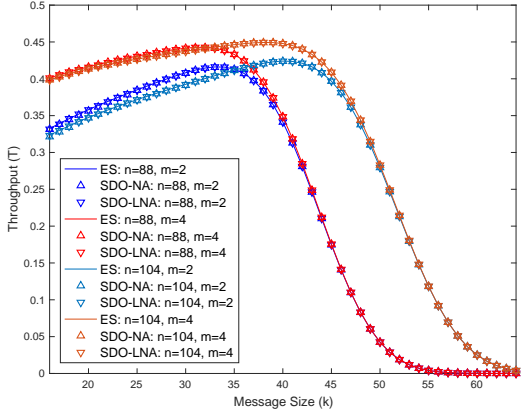}}\vspace{-0.25cm}
\caption{A comparison between exhaustive search (ES), the SDO optimization based on the normal approximation (SDO-NA) and the log-normal approximation (SDO-LNA).
While SDO-based methods are more computationally efficient, ensuing performance is nearly indistinguishable.}
\label{fig:TvsKESNA}\vspace{-0.25cm}
\end{figure}
  
Figure~\ref{fig:TvsN} shows the effect of the number of sub-blocks ($m$) on the throughput ($T$) as a function of the blocklength ($n$), for $k=32$ and $\epsilon=0.5$.
When $n$ is fixed, the throughput increases with $m$; this is to be expected.
Still, the benefits in terms of throughout for $m \geq 5$ becomes relatively small, offering diminishing returns.
This suggests that a small number of feedback messages suffice to achieve a throughput close to the maximum throughput obtained with unlimited feedback, an encouraging result for pragmatic systems.
\begin{figure}
\centering{\includegraphics[width=0.45\textwidth]{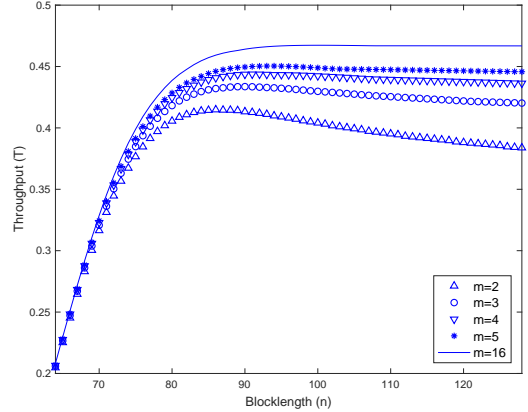}}\vspace{-0.25cm}
\caption{Throughput ($T$) as a function of blocklength ($n$) for different constraints on the maximum number of sub-blocks ($m$).
The optimal blocklength is very robust to constraint $m$, and largely dictated by channel characteristics.}
\label{fig:TvsN}\vspace{-0.25cm}
\end{figure}

More generally, the methodology developed in this work offers a technical pathway to using incremental redundancy for queueing analysis~\cite{ieee-tit-2013-kcp,ieee-tcomm-2015-hcp} or age of information problems~\cite{ephremides2016age,yates2017age} over erasure channels.
In addition, SDO may also play a role in real-time scheduling~\cite{hou2012scheduling}, fundamental limits of control under communication constraints~\cite{tatikonda2005feedback}, and learning methods for communication systems.

\appendix

\subsection{Proof of Lemma~\ref{lem:DSP}}
Let $H$ be a random matrix of size $(n-k) \times n$, where each entry is selected independently and uniformly from $\{0,1\}$. Consider an $(n,k)$ code with parity-check matrix $H$. Then, for any codeword $\mathbf{c}$, we have $H\mathbf{c}^{T} = 0$.
As such, the receiver can decode the message from any $r$ received symbols $c_{i_1},\dots,c_{i_r}$ provided that the $n-r$ columns of $H$ with indices $[n]\setminus \{i_1,\dots,i_r\}$ are linearly independent.
That is, $P_{\mathrm{s}}(k,n,r)$ is equal to the probability that $n-r$ randomly generated binary column vectors of length $n-k$ are linearly independent.
This event has probability
\begin{equation*}
P_{\mathrm{s}}(k,n,r) = \prod_{l=1}^{n-r}\frac{\left( 2^{n-k}-2^{l-1} \right)}{2^{n-k}}
= \prod_{l = 0}^{n-r-1} \left(1 - 2^{l-(n-k)} \right)
\end{equation*}
for $k\leq r\leq n$, and $P_{\mathrm{s}}(k,n,r) = 0$ for $r<k$.

\subsection{Proof of Lemma~\ref{lemma:StochaticDominance}}

Fix $k$ and $r$.
Consider two different blocklengths $n_1$ and $n_2$ such that $n_1 < n_2$.
If $r < k$, then the receiver cannot decode and $P_{\mathrm{s}} (k, n_1, r) = P_{\mathrm{s}} (k, n_2, r) = 0$.
By definition, when $r > n_1$, we have $P_{\mathrm{s}} (k, n_1, r) = 1$, which is necessarily greater than or equal to $P_{\mathrm{s}} (k, n_2, r)$.
Thus, the case of interest is $k \leq r \leq n_1 < n_2$, with
\begin{equation*}
\begin{split}
P_{\mathrm{s}} & (k, n_2, r)
= \prod_{l=r-k+1}^{n_2-r-1} \left( 1 - 2^{-l} \right) \\
&= \left( \prod_{l=r-k+1}^{n_1-r-1} \left( 1 - 2^{-l} \right) \right)
\times \left( \prod_{l=n_1-r}^{n_2-r-1} \left( 1 - 2^{-l} \right) \right) \\
&= P_{\mathrm{s}} (k, n_1, r)
\times \prod_{l=n_1-r}^{n_2-r-1} \left( 1 - 2^{-l} \right) \\
&\leq P_{\mathrm{s}} (k, n_1, r) .
\end{split}
\end{equation*}
By the arguments above, for $k$ and $r$ fixed, the function $P_{\mathrm{s}}(k,n,r)$ is monotone decreasing in $n$.

\subsection{Proof of Lemma~\ref{lem:infsums}}
First, we prove $\sum_{i=0}^{\infty} a_i = 1$.
Let \[S_l\triangleq \sum_{i=0}^{\infty} 2^{-i}\prod_{j=i+1}^{i+l}(1-2^{-j})\] for all $l\geq 1$, and let $S_0 \triangleq \sum_{i=0}^{\infty} 2^{-i}$.
Note that $\sum_{i=0}^{\infty}a_i=\lim_{l\rightarrow\infty} S_{l}$.
By some algebra, it follows that \[S_l = S_{l-1}-\frac{2^l}{(2^l-1)(2^{l+1}-1)}.\]
Then, it is easy to see that 
\begin{align*}
S_l & = S_0-\sum_{i=1}^{l} \frac{2^i}{(2^i-1)(2^{i+1}-1)}\\ & = S_0-\sum_{i=0}^{l}\frac{1}{2^i-1}+\sum_{i=0}^{l}\frac{1}{2^{i+1}-1}.	
\end{align*}
Thus, \[\lim_{l\rightarrow\infty} S_l = S_0-\sum_{i=1}^{\infty}\frac{1}{2^i-1}+\sum_{i=1}^{\infty}\frac{1}{2^{i+1}-1}.\]
Since $\sum_{i=1}^{\infty}\frac{1}{2^i-1}=c_0$ and $\sum_{i=1}^{\infty}\frac{1}{2^{i+1}-1}=c_0-1$, then 
\begin{align*}
 \lim_{l\rightarrow\infty} S_l = S_0-c_0+(c_0-1)=S_0-1.	
\end{align*}
Moreover, $S_0 = 2$.
Thus, \[\sum_{i=0}^{\infty}a_i=\lim_{l\rightarrow\infty} S_{l}=1.\] 


Next, we prove that $\sum_{i=0}^{\infty} ia_i= c_0$.
By using the Euler's pentagonal number theorem, it can be shown that 
\begin{equation}\label{eq:EulerPent}
\prod_{i=0}^{\infty} \frac{1}{1-2^{-i} x}=\sum_{i=0}^{\infty} x^i \prod_{j=1}^{i} \frac{1}{1-2^{-j}}. 	
\end{equation} Taking derivative with respect to $x$ from both sides of this identity, we get 
\begin{dmath*}
\left(\sum_{i=0}^{\infty}\frac{2^{-i}}{1-2^{-i}x}\right)\left(\prod_{j=0}^{\infty}\frac{1}{1-2^{-j}x}\right)=\sum_{i=1}^{\infty} ix^{i-1} \prod_{j=1}^{i} \frac{1}{1-2^{-j}}.	
\end{dmath*} Setting $x=1/2$, we get \begin{dmath*}
 \left(\sum_{i=0}^{\infty}\frac{2^{-i}}{1-2^{-i-1}}\right)\left(\prod_{j=0}^{\infty}\frac{1}{1-2^{-j-1}}\right)=\sum_{i=1}^{\infty} i2^{-i+1} \prod_{j=1}^{i} \frac{1}{1-2^{-j}}.	
\end{dmath*}
By a simple change of variables and rearranging the terms, 
\begin{dmath*}
\sum_{i=1}^{\infty}\frac{2^{-i}}{1-2^{-i}}=\sum_{i=1}^{\infty} i2^{-i} \prod_{j=i+1}^{\infty} (1-2^{-j}). 	
\end{dmath*}
Since 
\begin{dmath*}
\sum_{i=1}^{\infty} i2^{-i} \prod_{j=i+1}^{\infty} (1-2^{-j})=\sum_{i=0}^{\infty} ia_i, 	
\end{dmath*} then 
\begin{dmath*}
\sum_{i=0}^{\infty} ia_i=\sum_{i=1}^{\infty}\frac{2^{-i}}{1-2^{-i}}	
\end{dmath*}
Moreover, 
\begin{dmath*}
\sum_{i=1}^{\infty}\frac{2^{-i}}{1-2^{-i}}=\sum_{i=1}^{\infty} \frac{1}{2^{i}-1}=c_0. 	
\end{dmath*} Thus, $\sum_{i=0}^{\infty}ia_i=c_0$. 

Finally, we prove that $\sum_{i=0}^{\infty} i^{2} a_i = c_0^2+c_0+c_1$.
Taking derivative twice with respect to $x$ from both sides of identity~\eqref{eq:EulerPent} and setting $x=1/2$, we get 
\begin{dmath*}
\left(\sum_{i=1}^{\infty}(2^i-1)^{-1}\right)^2+\left(\sum_{i=1}^{\infty}(2^i-1)^{-2}\right)=\sum_{i=0}^{\infty} i(i-1)2^{-i} \prod_{j=i+1}^{\infty} (1-2^{-j}).	
\end{dmath*} Since $\sum_{i=1}^{\infty}(2^i-1)^{-1}=c_0$ and $\sum_{i=1}^{\infty}(2^i-1)^{-2}=c_1$, then 
\begin{dmath*}
\sum_{i=0}^{\infty} i(i-1)2^{-i} \prod_{j=i+1}^{\infty} (1-2^{-j}) = c_0^2+c_1. 	
\end{dmath*} Moreover, 
\begin{dmath*}
\sum_{i=0}^{\infty} i(i-1)2^{-i} \prod_{j=i+1}^{\infty} (1-2^{-j})=\sum_{i=0}^{\infty} i^2a_i -\sum_{i=0}^{\infty} ia_i. 	
\end{dmath*} Thus, 
\begin{dmath*}
\sum_{i=0}^{\infty} i^2a_i -\sum_{i=0}^{\infty} ia_i=c_0^2+c_1. 	
\end{dmath*} Since $\sum_{i=0}^{\infty} ia_i=c_0$, then $\sum_{i=0}^{\infty} i^2a_i=c_0^2+c_0+c_1$.

\bibliographystyle{IEEEtran}
\bibliography{isit2018}

\begin{thebibliography}{10}
\providecommand{\url}[1]{#1}
\csname url@samestyle\endcsname
\providecommand{\newblock}{\relax}
\providecommand{\bibinfo}[2]{#2}
\providecommand{\BIBentrySTDinterwordspacing}{\spaceskip=0pt\relax}
\providecommand{\BIBentryALTinterwordstretchfactor}{4}
\providecommand{\BIBentryALTinterwordspacing}{\spaceskip=\fontdimen2\font plus
\BIBentryALTinterwordstretchfactor\fontdimen3\font minus
  \fontdimen4\font\relax}
\providecommand{\BIBforeignlanguage}[2]{{%
\expandafter\ifx\csname l@#1\endcsname\relax
\typeout{** WARNING: IEEEtran.bst: No hyphenation pattern has been}%
\typeout{** loaded for the language `#1'. Using the pattern for}%
\typeout{** the default language instead.}%
\else
\language=\csname l@#1\endcsname
\fi
#2}}
\providecommand{\BIBdecl}{\relax}
\BIBdecl

\bibitem{polyanskiy2010dispersion}
Y.~Polyanskiy, H.~V. Poor, and S.~Verd{\'u}, ``Channel coding rate in the
  finite blocklength regime,'' \emph{IEEE Trans.\ Inf.\ Theory}, vol.~56,
  no.~5, pp. 2307--2359, May 2010.

\bibitem{polyanskiy2011dispersion}
------, ``Dispersion of the {Gilbert-Elliott} channel,'' \emph{IEEE Trans.\
  Inf.\ Theory}, vol.~57, no.~4, pp. 1829--1848, April 2011.

\bibitem{ieee-tit-2013-kcp}
S.~Kumar, J.-F. Chamberland, and H.~D. Pfister, ``First-passage time and
  large-deviation analysis for erasure channels with memory,'' \emph{IEEE
  Trans.\ Inf.\ Theory}, vol.~59, no.~9, pp. 5547--5565, September 2013.

\bibitem{ephremides2016age}
M.~Costa, M.~Codreanu, and A.~Ephremides, ``On the age of information in status
  update systems with packet management,'' \emph{IEEE Trans.\ Inf.\ Theory},
  vol.~62, no.~4, pp. 1897--1910, April 2016.

\bibitem{yates2017age}
Y.~Sun, E.~Uysal-Biyikoglu, R.~D. Yates, K.~C. E., and N.~B. Shroff, ``Update
  or wait: {How} to keep your data fresh,'' \emph{IEEE Trans.\ Inf.\ Theory},
  vol.~63, no.~11, pp. 7492--7508, August 2017.

\bibitem{le2007queueing}
L.~B. Le, E.~Hossain, and M.~Zorzi, ``Queueing analysis for {GBN} and {SR ARQ}
  protocols under dynamic radio link adaptation with non-zero feedback delay,''
  \emph{IEEE Trans.\ Wireless Commun.}, vol.~6, no.~9, pp. 3418--3428,
  September 2007.

\bibitem{ieee-tcomm-2015-hcp}
F.~Hamidi-Sepehr, J.-F. Chamberland, and H.~D. Pfister, ``On the performance of
  block codes over finite-state channels in the rare-transition regime,''
  \emph{IEEE Trans.\ Commun.}, vol.~63, no.~11, pp. 3974--3990, November 2015.

\bibitem{wesel2015harq}
A.~R. Williamson, T.-Y. Chen, and R.~D. Wesel, ``Variable-length convolutional
  coding for short blocklengths with decision feedback,'' \emph{IEEE Trans.\
  Commun.}, vol.~63, no.~7, pp. 2389--2403, May 2015.

\bibitem{VWRDW:2014}
K.~Vakilinia, A.~R. Williamson, S.~V.~S. Ranganathan, D.~Divsalar, and R.~D.
  Wesel, ``Feedback systems using non-binary {LDPC} codes with a limited number
  of transmissions,'' in \emph{IEEE Information Theory Workshop (ITW)}, Hobart,
  Australia, Nov. 2014.

\bibitem{VRDW:2016}
K.~Vakilinia, S.~V.~S. Ranganathan, D.~Divsalar, and R.~D. Wesel, ``Optimizing
  transmission lengths for limited feedback with non-binary ldpc examples,''
  \emph{IEEE Trans.\ Commun.}, vol.~64, no.~6, pp. 2245--2257, Jun. 2016.

\bibitem{Gallager0471290483}
R.~G. Gallager, \emph{Information Theory and Reliable Communication}.\hskip 1em
  plus 0.5em minus 0.4em\relax Wiley, 1968.

\bibitem{Richardson0521852293}
T.~Richardson and R.~Urbanke, \emph{Modern Coding Theory}.\hskip 1em plus 0.5em
  minus 0.4em\relax Cambridge University Press, 2008.

\bibitem{casella2001statistical}
G.~Casella and R.~L. Berger, \emph{Statistical Inference}, 2nd~ed.\hskip 1em
  plus 0.5em minus 0.4em\relax Duxbury Thomson Learning, 2001.

\bibitem{hou2012scheduling}
I.-H. Hou and P.~R. Kumar, ``Real-time communication over unreliable wireless
  links: a theory and its applications,'' \emph{IEEE Wireless Commun.\ Mag.},
  vol.~19, no.~1, pp. 48--59, February 2012.

\bibitem{tatikonda2005feedback}
S.~Yang, A.~Kavcic, and S.~Tatikonda, ``Feedback capacity of finite-state
  machine channels,'' \emph{IEEE Trans.\ Inf.\ Theory}, vol.~51, no.~3, pp.
  799--810, February 2005.

\end{thebibliography}

\end{document}